\documentclass{article}

\usepackage{arxiv}

\usepackage[utf8]{inputenc} 
\usepackage[T1]{fontenc}    
\usepackage{hyperref}       
\usepackage{url}            
\usepackage{booktabs}       
\usepackage{amsfonts}       
\usepackage{nicefrac}       
\usepackage{microtype}      
\usepackage{lipsum}

\usepackage{cite}
\usepackage{amsmath,amssymb,amsfonts}
\usepackage{graphicx}
\usepackage{algorithm}
\usepackage[noend]{algpseudocode}
\usepackage{hyperref}
\usepackage{textcomp}
\usepackage{threeparttablex} 

\usepackage{mathtools}
\usepackage{caption}
\usepackage{float}
\usepackage{stmaryrd}
\usepackage{epstopdf}
\usepackage{multirow}
\usepackage{pifont}
\usepackage{caption}
\usepackage{subcaption}

\usepackage{tikz}
\usetikzlibrary{automata, positioning, arrows}

\usepackage{caption}
\usepackage{subcaption}

\newcommand{\R}{\mathbb{R}}

\newtheorem{theorem}{Theorem}[section]

\newtheorem{assumption}{Assumption}

\newtheorem{definition}[theorem]{Definition}
\newtheorem{lemma}[theorem]{Lemma}
\newtheorem{remark}[theorem]{Remark}
\newtheorem{problem}[theorem]{Problem}
\newtheorem{proof}[theorem]{Proof}

\title{Spatiotemporal Tubes based Controller Synthesis against Omega-Regular Specifications for Unknown Systems
\thanks{The work is supported in part by the Kotak-IISc AI/ML Center.}
}

\author{
 Aiman Aatif Bayezeed\thanks{Authors contributed equally.} \\
  Centre for Cyber-Physical Systems\\
  IISc, Bengaluru, India\\
  \texttt{aimanatifb@iisc.ac.in} \\
   \And
 Ratnangshu Das$^\dag$ \\
  Centre for Cyber-Physical Systems\\
  IISc, Bengaluru, India\\
  \texttt{ratnangshud@iisc.ac.in} \\
  \And
 Pushpak Jagtap \\
  Centre for Cyber-Physical Systems\\
  IISc, Bengaluru, India\\
  \texttt{pushpak@iisc.ac.in} \\
}

\begin{document}
\maketitle

\begin{abstract}
This paper provides a discretization-free solution to the synthesis of approx-imation-free closed-form controllers for unknown nonlinear systems to enforce complex properties expressed by $\omega$-regular languages, as recognized by Non-deterministic Büchi Automata (NBA). In order to solve this problem, we first decompose NBA into a sequence of reach-avoid problems, which are solved using the Spatiotemporal Tubes (STT) approach. Controllers for each reach-avoid task are then integrated into a hybrid policy that ensures the fulfillment of the desired $\omega$-regular properties. We validate our method through omnidirectional robot navigation and manipulator control case studies.
\end{abstract}

\section{Introduction}
\label{sec1}

In recent years, the need to address complex specifications has gained increasing importance in the fields of robotics, autonomous systems, and control theory. The specifications expressed using temporal logic formulae \cite{pnueli1977temporal}  or (in)finite strings over automata \cite{thomas1990automata} provide a powerful and expressive framework for formally describing system behaviors. As robotic systems grow more complex and capable of executing long-term tasks, designing controllers that satisfy such specifications is essential for applications in autonomous navigation, multi-agent coordination, and safety-critical systems.

To tackle these complex tasks, abstraction-based techniques \cite{tabuada2009verification} have been widely employed, where the continuous dynamics of a system is abstracted into a finite abstraction using state and input discretization. However, this approach suffers from significant computational burdens, especially in high-dimensional systems, where the state space discretization results in an exponential growth in the number of states, making it computationally impractical for real-time control. In addition, they rely on the exact knowledge of the dynamics.

To overcome the scalability of abstraction-based techniques, recent approaches have explored discretization-free methods using control barrier functions (CBFs) \cite{jagtap2020formal, anand2021compositional,sundarsingh2023scalable}. CBFs provide a way to synthesize controllers that guarantee system safety while satisfying temporal logic constraints. However, designing an appropriate CBF for a system is challenging. Existing methods typically rely on parametric forms of CBFs \cite{jagtap2020formal, anand2021compositional, ramasubramanian2019linear}, and involve computationally intensive numerical searches, which limit scalability, particularly when dealing with complex systems and specifications. Some studies have also explored optimization-based CBF methods for temporal-logic motion planning \cite{he2020bp, srinivasan2020control}. However, these approaches do not provide closed-form controllers and often depend on computational methods
that are numerically intractable and lack feasibility guarantees. Moreover, CBF-based techniques require precise system dynamics, which can be a serious constraint for real-world systems.

An alternative approach is funnel-based control \cite{bechlioulis2014low}, which helps synthesize closed-form hybrid control policies to enforce complex temporal specifications. Previous work, such as \cite{lindemann2017prescribed, lindemann2019feedback}, addressed fragments of signal temporal logic (STL), while \cite{verginis2018timed} focuses on metric interval temporal logic (MITL) for cooperative manipulation using state-space discretization. However, these results are limited to fragments of temporal logic and do not apply to more general omega-regular languages \cite{thomas1990automata}. In \cite{NAHS}, a funnel-based approach was proposed to decompose temporal logic specifications, recognized by $\omega$-regular languages or non-deterministic Büchi automata \cite{buchi1990decision}, into a sequence of reachability tasks and combine the resulting controllers into a hybrid control policy. However, this approach imposes restrictive assumptions that the avoid region must not intersect with the funnel for all time, which limits its applicability to more realistic scenarios involving nontrivial avoid specifications. Additionally, this result was derived under the assumption of known system dynamics with a control-affine structure, limiting its applicability.

To address these challenges, this paper introduces a spatiotemporal tubes (STT)-based approach \cite{STT} to synthesize approximation-free and closed-form hybrid control policies that satisfy $\omega$-regular specifications recognized by non-deterministic Büchi automata. This method provides more flexibility, addressing a broader range of reach-avoid tasks. By decomposing the NBA into a sequence of reach-avoid problems, each task is solved using STTs, ensuring robust safety and temporal guarantees. The controllers are then integrated into a hybrid policy that ensures that the system adheres to the desired omega-regular specifications, offering a scalable and efficient solution for complex real-world systems.
The key contributions of the paper can be summarized as follows.
\begin{enumerate}
    \item The proposed approach is applicable to more general strict feedback Multiple-Input Multiple-Output (MIMO) systems with unknown dynamics compared to \cite{NAHS}.
    \item In contrast to \cite{NAHS}, the proposed approach is capable of addressing a broad range of $\omega$-regular specifications by relaxing assumptions on the positions of the obstacles.
    \item In contrast to \cite{jagtap2020formal}, \cite{anand2021compositional}, and  \cite{NAHS}, the proposed approach provides a closed-form control law that is robust and approximation-free (i.e., it is capable of addressing systems with unknown dynamics subjected to unknown bounded disturbances).
\end{enumerate}

The paper is organized as follows. Section \ref{sec2} introduces the system dynamics and the $\omega$-regular property specifications, followed by the problem definition. Section \ref{sec3} outlines the mapping of $\omega$-regular specifications into a sequence of reach-avoid tasks and their resolution through a spatiotemporal tubes-based hybrid control framework. Section \ref{sec5} demonstrates the approach's practical applicability through case studies on 2R robot manipulator control and omnidirectional robot navigation, showcasing its computational efficiency and effectiveness. Finally, Section \ref{conclusion} concludes the paper. 

\section{Preliminaries and Problem Statement}
\label{sec2}

\subsection{Notations}
\label{subsec2_1}

The symbols $\mathbb{N}, \mathbb{N}_0, \mathbb{Z}, \mathbb{R}, \mathbb{R}^+,$ and $\mathbb{R}_0^+$ denote the set of natural, nonnegative integer, integer, real, positive real, and nonnegative real numbers, respectively. We use $\mathbb{R}^{n\times m}$ to denote a vector space of real matrices with $n$ rows and $m$ columns. We use $\left \lVert \cdot \right \rVert$ to represent the Euclidean norm. For $a,b \in \mathbb{R}$ and $a < b$, we use $(a,b)$ to represent the open interval in $\mathbb{R}$. For $a,b \in \mathbb{N}$ and to represent the Euclidean norm. We use $I_n$ and $0_{n\times m}$ to denote the identity matrix in $\mathbb{R}^{n \times n}$ and the zero matrix in $\mathbb{R}^{n \times m}$, respectively. A diagonal matrix in $\mathbb{R}^{n \times n}$ with diagonal entries $d_1, ..., d_n$ is denoted by $diag \{ d_1, ... , d_n \}$. Given a matrix $M \in \mathbb{R}^{n \times m},$ $ M^T$ represents the transpose of matrix $M$. Given a matrix $P \in \mathbb{R}^{n \times n}$, Tr($P$) represents trace of matrix $P$. Given a set $A$, we use $|A|$ to represent the cardinality of the set $A$. $x_i, i \in [1;n]$ denoted the $i$-th element of the vector $x \in \mathbb{R}^n$. Consider $N$ sets $A_i, i \in \{ 1,...,N \}$, the Cartesian product of the sets is given by $A = \Pi_{i \in \{a_1, ..., a_N\}} A_i := \{ (a_1, ..., a_N) \, | \, [ x_1, x_2, ..., x_n ] \in X_a \};$ and $Int(X_a)$ represents the interior of the set $X_a$. We denote the empty set by $\phi$. Given a set $S$, the notation $|S|$ denotes the cardinality of $S$; $S^*$ and $S^\omega$ denote the set of all finite and infinite strings over $S$, respectively. Given a set $U$ and $S \subset U$, the complement of $S$ with respect to $U$ is defined as $U \setminus S = \{ x : x \in U, x \notin S \}$.

\subsection{System Description}
\label{subsec2_2}

Consider a class of control-affine MIMO nonlinear pure-feedback systems characterized by the following dynamics:
\begin{align} \label{eqn:sysdyn}
    \mathcal{S}: \left\{\begin{matrix}\dot{x}_k(t) = f_k(\overline{x}_k(t)) + g_k(\overline{x}_k(t))x_{k+1}(t) + w_k(t), k\in [1;N-1], \hfill\\
    \dot{x}_{N}(t) = f_N(\overline{x}_N(t)) + g_N(\overline{x}_N(t))u(t) + w_N(t), \hfill\\
    y(t) = x_1(t),\hfill\end{matrix}\right.
\end{align}
where for $t\in\R^+_0$ and $i\in[1;N]$,
\begin{itemize}
    \item $x_k(t) = [x_{k,1}(t), \ldots, x_{k,n}(t)]^\top \in {X}_k \subset \mathbb{R}^{n}$ is the state vector, 
    \item $\overline{x}_k(t) := [x_1^\top(t),x_2^\top(t),...,x_k^\top(t)]^\top \in \overline{X}_k = \prod_{j=1}^k X_j \subset \mathbb{R}^{ni} $, 
    \item $u(t) \in \mathbb{R}^n$ is the control input vector,
    \item $w_k(t) \in W \subset \R^n$ is the unknown bounded external disturbance, and
    \item $y(t) = [y_1(t), \ldots, y_n(t)] = [x_{1,1}(t), \ldots, x_{1,n}(t)]\in Y=X_1$ denotes the output vector.
\end{itemize}
We denote the value of an output trajectory of $\mathcal{S}$ starting from $y_0 \in Y$ under a control input signal $u$ at time $t\in\R^+_0$ by $y_{y_0u}(t)$.
The functions $f_k: \overline{X}_k \rightarrow \mathbb{R}^n$, $g_k: \overline{X}_k \rightarrow \mathbb{R}^{n \times n}, k \in [1;N]$, follow the assumptions \ref{assum:lip} and \ref{assum:pd}.

\begin{assumption}\label{assum:lip}
    For all $k \in [1;N]$, functions $f_k$ and $g_k$ are unknown and locally Lipschitz.
\end{assumption}
\begin{assumption}[\cite{PPC1,PPC0}] \label{assum:pd}
    The matrix $g_{k,s}(\overline{x}_i) = \frac{g_k(\overline{x}_k)+g_k(\overline{x}_k)^\top}{2}$ is uniformly sign definite with known signs for all $\overline{x}_k \in \overline{X}_k$. Without loss of generality, one can assume that $g_{k,s}(\overline{x}_k)$ is positive definite, that is, for all $\overline{x}_k \in \overline{X}_k,$ there exists a constant $\underline{g_k}\in\mathbb R^+, \forall k \in [1;N]$ such that
    $$0 < \underline{g_k} \leq \lambda_{\min} (g_{k,s}(\overline{x}_k)),$$
    where $\lambda_{\min}(\cdot)$ represents the smallest eigenvalue of the matrix.
    This assumption establishes a global controllability condition for \eqref{eqn:sysdyn}.
\end{assumption}

\subsection{Class of Specifications}
\label{subsec2_3}
This work aims to synthesize controllers for unknown systems \eqref{eqn:sysdyn} to guarantee the satisfaction of $\omega$-regular properties, which are commonly used to describe complex, infinite-horizon behaviors in dynamical systems. These properties are recognized by various types of automata that accept infinite words, including non-deterministic B{\"u}chi automata (NBA) \cite{buchi1990decision}, deterministic Rabin automata \cite{rabin1969decidability}, deterministic Streett automata \cite{streett1982propositional}, parity automata or Muller automata \cite{muller1963infinite}.  Despite differences in acceptance conditions, they share the same expressive power in recognizing $\omega$-regular languages. In this work, we focus on using nondeterministic Büchi automata (NBA) to describe $\omega$-regular properties in control systems.

\begin{definition}[\cite{buchi1990decision}]\label{Def:NBA}
    A nondeterministic B{\"u}chi automaton (NBA) is defined as a tuple $\mathcal{A} = (\mathcal{Q},\mathcal{Q}_0, \Sigma, \delta, F)$, where $\mathcal{Q}$ is a finite set of states, $\mathcal{Q}_0 \subseteq \mathcal{Q}$ is the set of initial states, $\Sigma$ is a finite alphabet, $\delta : \mathcal{Q} \times \Sigma \rightarrow 2^\mathcal{Q}$ is a transition function, and $F \subseteq \mathcal{Q}$ is the set of accepting states.
\end{definition}

A transition $(q,\sigma,q') \in \delta$ in the NBA $\mathcal{A}$ is denoted as $q \xrightarrow{\sigma}_\mathcal{A} q'$. Consider an infinite state-run, $q = (q_0, q_1,...) \in \mathcal{Q}^\omega$ and an infinite word (a.k.a. trace) $\sigma = (\sigma_0,\sigma_1,...) \in \Sigma^\omega$. The run begins at an initial state, $q_0 \in \mathcal{Q}_0$, and for each $i \in \mathbb{N}_0$, the transition $q_i \xrightarrow{\sigma_i}_\mathcal{A} q_{i+1}$ holds. Let the set of states that appear infinitely often in the run be denoted by Inf(q). An infinite word $\sigma = (\sigma_0, \sigma_1,...) \in \Sigma^\omega$ is accepted by NBA $\mathcal{A}$ if there exists an infinite state run q corresponding to $\sigma$ such that Inf(q) $\cap F \neq \phi$. The set of words accepted by $\mathcal{A}$ is called the accepting language of $\mathcal{A}$, denoted by $\mathcal{L}(\mathcal{A})$.
We consider specifications defined by the accepting languages of NBA $\mathcal{A}$, where the alphabets are defined over the set of atomic propositions $\Pi$, i.e., $\Sigma = 2^\Pi$. 

\begin{remark}
    The approach proposed in this paper can handle $\omega$-regular languages represented by deterministic Streett automata and regular languages expressed via (non)deterministic finite automata (NFA) \cite{baier2008principles}. From the perspective of temporal logic, nondeterministic B{\"u}chi automata and deterministic Streett automata can capture the full range of linear temporal logic (LTL) properties with conversions facilitated by tools such as SPOT \cite{duret2016spot}, LTL2BA \cite{gastin2001fast}, and ltl2dstar \cite{klein2007ltl2dstar}. The LTL specifications over finite traces (LTL$_f$) \cite{de2013linear} can be translated into deterministic finite automata (DFA) using tools like MONA \cite{henriksen1995mona}.
\end{remark}

\subsection{Satisfaction of Specification by Systems}
\label{subsec2_4}

Let a system $\mathcal{S}$ as defined in \eqref{eqn:sysdyn} be assigned a specification described by the accepting language of an NBA $\mathcal{A}$. These specifications are defined over a set of atomic propositions $\Pi$, and the connection between the system and the specification is established through a labeling function $L: Y \rightarrow \Pi$. This function maps the output of the system $y \in Y$ to the corresponding atomic proposition, translating the system's behavior into a trace over $\Pi$.

\begin{definition}[\cite{wongpiromsarn2015automata}]\label{Def:Trace_of_system}
    For a given output trajectory $y_{y_0u}$ of the system $\mathcal{S}$ in \eqref{eqn:sysdyn} and a labeling function $L : Y \rightarrow \Pi$, the infinite sequence $\sigma(y_{y_0u}) = (\sigma_0, \sigma_1, ...) \in \Pi^\omega$ is called an infinite trace of the output trajectory $y_{y_0u}$, if there exists a time sequence $t_0, t_1, ...$ with $t_0 = 0, t_r \rightarrow \infty$ as $r \rightarrow \infty$, such that for all $j \in \mathbb{N}, t_j \in \mathbb{R}_0^+$: 
\begin{itemize}
    \item $t_j < t_{j+1};$
     \item The output of the system at a time belongs to the pre-image of the corresponding label, i.e., $y_{y_0u}(t_j) \in L^{-1}(\sigma_j);$
    \item If $\sigma_j \neq \sigma_{j+1},$ then for some $t'_j \in [t_j, t_{j+1}],$ $ y_{y_0u}(t) \in L^{-1}(\sigma_j)$ for all $t \in (t_j, t'_j)$; $y_{y_0u}(t) \in L^{-1}(\sigma_{j+1})$ for all $t \in (t'_j, t_{j+1})$; and either $y_{y_0u}(t'_j) \in L^{-1}(\sigma_j)$ or $y_{y_0u}(t'_j) \in L^{-1}(\sigma_{j+1})$.
\end{itemize} 
\end{definition}
 
Next, we formally define the satisfaction of specification given by the language of $\mathcal{A}$ by the output trajectory of the system $ \mathcal{S}$.

\begin{definition}
    The output trajectory $y_{y_0u}$ of system $ \mathcal{S}$ in \eqref{eqn:sysdyn}, starting from $y_0 \in Y$ and under an input signal $u$, satisfies the specification given by the accepting language of an NBA $\mathcal{A}$, denoted by $\sigma(y_{y_0u}) \models \mathcal{A}$, if the infinite trace generated by its trajectory $\sigma(y_{y_0u})$ as defined in Definition \ref{Def:Trace_of_system} satisfies $\sigma(y_{y_0u}) \in \mathcal{L}(\mathcal{A})$.
\end{definition}

\subsection{Problem Definition}
\label{subsec2_5}
The controller synthesis problem addressed in this work is defined below.
\begin{problem}\label{problem}
    Consider an unknown system $\mathcal{S}$ in \eqref{eqn:sysdyn} satisfying the assumptions \ref{assum:lip} and \ref{assum:pd}, a specification given by the accepting language of an NBA $\mathcal{A} = (\mathcal{Q},\mathcal{Q}_0, \Sigma, \delta, F)$ defined over a set of atomic propositions $\Pi = \{p_0,p_1, ..., p_M \}$, and a labeling function $L: Y \rightarrow \Pi$. The aim is to design an approximation-free closed-form hybrid controller $u$ such that $\sigma(y_{y_0u}) \vDash \mathcal{A}$ for all $y_0\in L^{-1}(p_i)$ for some $i \in \{1,\ldots, M\}$ $($Note that the controller will depend on the initial state, meaning different controllers will be designed for different starting regions corresponding to each initial atomic proposition$)$.
\end{problem}

To solve this problem, we first decompose the specification given by the language of NBA $\mathcal{A}$ into a sequence of simpler reach-avoid specifications. Then, we introduce a spatiotemporal tube-based control framework \cite{STT} to solve these reach-avoid specifications.

\section{Reach-Avoid tasks using Spatiotemporal Tubes}
\label{sec3}
{In this section, we present an approach inspired by \cite{NAHS} to decompose complex specifications, expressed in the language of NBA, into a sequence of simpler reach-avoid (RA) tasks. This decomposition enables the handling of intricate $\omega$-regular requirements by breaking them into manageable sub-tasks. Each RA task is then addressed using the spatiotemporal tubes-based method presented in \cite{STT}. Finally, we propose a closed-form, approximation-free hybrid control policy that is capable of satisfying the original specifications defined by the language of an NBA.}

\subsection{Decomposing Specifications to sequence of Reach-Avoid Tasks}
\label{subsec_mapping}

{To decompose an NBA into a sequence of RA tasks, we begin by considering the NBA $\mathcal{A} = (\mathcal{Q}, \mathcal{Q}_0, \Pi, \delta, F)$, which corresponds to an $\omega$-regular language that defines the properties of interest for the system $\mathcal{S}$.} \\ 

{\textbf{Accepting State Runs.}
Given an accepting run of the NBA $q$, representing a sequence of states that satisfy the specification, the corresponding infinite words are denoted by $\sigma$(q) $\subseteq$ $\Pi^\omega$. Similarly, the finite words associated with the finite state runs are denoted by $\sigma(\overline{\text{q}}) \in \Pi^n$, where $ \overline{\text{q}} \in \mathcal{Q}^{n+1}$ and $n \in \mathbb{N}$. It is known that a word $\sigma \in \Pi^\omega$ is accepted by $\mathcal{A}$ if there exists a state run of the form q $= (q_0^r,q_1^r, ..., q_{m_r}^r, (q_0^s, q_1^s, ..., q_{m_s}^s)^\omega) \in \mathcal{Q}^\omega$, where $m_r, m_s \in \mathbb{N}, q_0^r \in \mathcal{Q}_0 $ and $q_0^s \in F$.} \\

{\textbf{Finite State-run Fragments.}
Let $\overline{\text{q}}$ be a finite state run fragment of an accepting run $q$, which is constructed by considering the infinite sequence $(q_0^s, q_1^s, ..., q_{m_s}^s)$ only once. This is expressed as $\overline{\text{q}} = (q_0^r,q_1^r, ..., q_{m_r}^r$, $q_0^s, q_1^s, ..., q_{m_s}^s, q_0^s, q_1^s) \in \mathcal{Q}^*$.
Let $\mathcal{R}$ represent the set of all such finite state-run fragments, excluding self-loops:
\begin{gather} \label{eqn4.1}
\mathcal{R} := \{ \overline{\text{q}} = (q_0^r,q_1^r, ..., q_{m_r}^r, q_0^s, q_1^s, ..., q_{m_s}^s, q_0^s, q_1^s) \; | \; q_0^r \in \mathcal{Q}_0, \nonumber \\ 
q_0^s \in F, q_i^r \neq q_{i+1}^r, \forall i < m_r, \text{ and } q_j^s \neq q_{j+1}^s, \forall j < m_s \}.
\end{gather}
The set $\mathcal{R}$ can be computed algorithmically by viewing $\mathcal{A}$ as a directed graph $\mathcal{G} = (\mathcal{V}, \mathcal{E})$ where the vertices $\mathcal{V} = \mathcal{Q}$ represent states, and the edges $\mathcal{E} \subseteq \mathcal{V} \times \mathcal{V}$ represent transitions between states. Specifically, $(q,q') \in \mathcal{E}$ if and only if $q' \neq q$ and there exist $p \in \Pi$ such that $q \xrightarrow{p}_\mathcal{A} q'$. A path in $\mathcal{G}$ is defined as a finite sequence of states $(q_0, q_1, ..., q_{\Tilde{n}}) \in \mathcal{Q}^{\Tilde{n}},$ for $\Tilde{n} \in \mathbb{N},$ satisfying $(q_i, q_{i+1}) \in \mathcal{E}$ for all $i \in [0, ..., \Tilde{n}-1]$. The atomic proposition associated with each edge $(q,q')$ is denoted by $\sigma (q,q')$. The variants of depth-first search can be used to compute $\mathcal{R}$ from $\mathcal{G}$ \cite{russell2003artificial}. It is worth noting that in the proposed approach, we do not need to find a complete set $\mathcal{R}$.}\\

{\textbf{Partitioning Based on Initial Atomic Propositions.}
For each $p \in \Pi$, we define the set $\mathcal{R}^p$ as:
\begin{equation} \label{eqn4.2}
    \mathcal{R}^p := \{ \overline{\text{q}} = (q_0^r,q_1^r, ..., q_{m_r}^r, q_0^s, q_1^s, ..., q_{m_s}^s, q_0^s, q_1^s) \in \mathcal{R} \; | \; \sigma(q_0^r,q_1^r) = p \}.
\end{equation}
This notation, with the superscript $p$, partitions of set $\mathcal{R}$ based on the initial atomic proposition of the finite state-run fragments defined earlier. This is particularly useful in designing hybrid controllers for each initial state set $L^{-1}(p_i), i \in \{ 1,2, ..., M \}$.} \\

{\textbf{Decomposition into Triplets.}
The decomposition is done by breaking down each state-run fragment into smaller segments. For any $\Bar{q} = (q_0,q_1, ..., q_{{m_r}+{m_s}+3}) \in \mathcal{R}^p$, we define $\mathcal{P}^p (\Bar{\text{q}})$ as the set of all state runs of length 3 or triplets, given by:
\begin{equation} \label{eqn4.3}
    \mathcal{P}^p (\Bar{\text{q}}) := \{ (q_i, q_{i+1}, q_{i+2}) \; | \; 0 \leq i \leq m_r + m_s +1 \}.
\end{equation}}

{We kindly refer the interested reader to Example 1 in \cite{NAHS}, which illustrates the process of generating triplets for a given NBA.}
\begin{figure}[ht]
    \centering
    \includegraphics[page={5}, width= 0.5\textwidth]{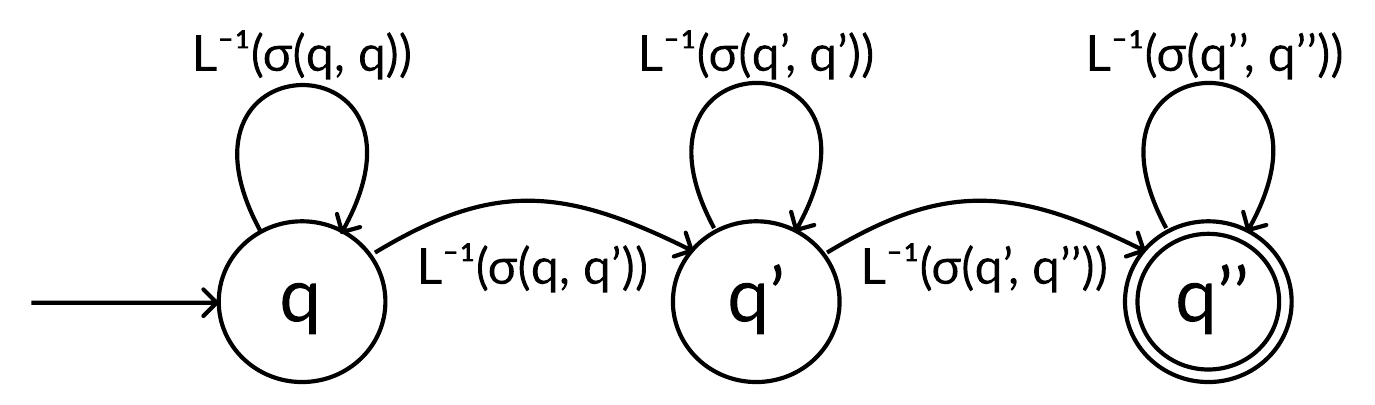}
    \caption{Illustrative automata representing triplets}
    \label{fig:example}
\end{figure}

{\textbf{Representing Triplets as RA tasks.}
As illustrated in Figure \ref{fig:example}, each triplet $\nu=(q,q',q'') \in \mathcal{P}^p (\overline{\text{q}})$ can be interpreted as a reach-avoid task, defined by:
\begin{itemize}
    \item Initial set $S := L^{-1}(\sigma(q,q'))$,
    \item Target set $T := L^{-1}(\sigma(q',q''))$,
    \item Unsafe set $U := Y \setminus \left( L^{-1}(\sigma(q',q')) \cup L^{-1}(\sigma(q,q')) \cup L^{-1}(\sigma(q',q'')) \right)$, where $Y = L^{-1}(\top) $.
\end{itemize}
This mapping translates each triplet into a control task, ensuring that the system starts in S, reaches T, and avoids U, facilitating systematic controller design.}

\subsection{Designing Spatiotemporal Tubes}
Given a reach-avoid task, that is, starting from set $S$, reach the target set $T$ while avoiding the unsafe set $U$, the spatiotemporal tubes (STTs) \cite{STT} are defined by the time-varying intervals $[\gamma_{i,L}(t), \gamma_{i,U}(t)]^\top$, where $\gamma_{i,L}:\mathbb{R}_0^+ \rightarrow \mathbb{R}$ and $\gamma_{i,U}:\R_0^+ \rightarrow \R$ are continuously differentiable functions with $\gamma_{i,L}(t) < \gamma_{i,U}(t)$ for all time $t \in \R_0^+$ and for each dimension $i \in [1;n]$, and the following hold:
\begin{align} \label{eqn:stt}
    &\prod_{i=1}^n [\gamma_{i,L}(0), \gamma_{i,U}(0)]^\top \subseteq S, \nonumber \\
    &\prod_{i=1}^n [\gamma_{i,L}(t_c), \gamma_{i,U}(t_c)]^\top \subseteq T, \text{ for some } t_c \in \R^+, \nonumber \\
    &\prod_{i=1}^n [\gamma_{i,L}(t), \gamma_{i,U}(t)]^\top \cap U = \emptyset, \forall t \in \R_0^+.
\end{align}

The construction of spatiotemporal tubes (STTs), as outlined in \cite{STT}, involves three sequential steps to satisfy the reach-avoid specifications:
\begin{enumerate}
    \item Design of Reachability Tubes: The process begins by constructing reachability tubes that guide the system's trajectory from the initial set $S$ to the target set $T$, ignoring the unsafe set $U$.
    \item Introduction of Circumvent Function \cite{das2024funnel}: To handle the avoid specifications, this step identifies time intervals during which the reachability tubes intersect with the unsafe set $U$. During this time interval, a circumvent function is defined, which acts as a spatiotemporal re-routing mechanism to steer the trajectory away from unsafe set $U$.
    \item Adaptive Tube Adjustment: In the final step, the framework leverages an adaptive mechanism to continuously adjust the reachability tubes in response to the circumvent function. The adaptive framework ensures that the adjusted tubes are modified around the circumvent function while maintaining proximity to the original reachability tube.
\end{enumerate}
Alternatively, a sampling-based approach can also be used, as outlined in \cite{das2024spatiotemporal}, to construct a spatiotemporal tube satisfying \eqref{eqn:stt}.

The reach-avoid specification can be enforced by designing a controller that constrains the output trajectory $y(t) = x_1(t)$ within the STTs as: 
\begin{gather}
    \gamma_{i,L}(t) < y_i(t) < \gamma_{i,U}(t), \forall (t,i) \in \R_0^+ \times [1;n]. \label{eqn:ppc}
\end{gather}

\subsection{Controller Design}
In this section, we derive an \textit{approximation-free}, \textit{closed-form} control law using STTs, similar to \cite{PPC2}, to satisfy \eqref{eqn:ppc}. The lower triangular structure of \eqref{eqn:sysdyn} allows us to use a backstepping-like design approach similar to that described in \cite{feedback}. First, we design an intermediate control input $r_2$ for the $y = x_1$ dynamics to ensure the fulfillment of \eqref{eqn:ppc}. We then iteratively design the intermediate control laws $r_{k+1}$ for the $x_k$ dynamics, ensuring $x_k$ tracks $r_k$, for all $k \in [2;N]$. It is important to note that $r_{N+1}$ effectively becomes the actual control input $u$ for the system. The steps of the controller design are outlined below.

\textbf{Stage $1$}: Given $\gamma_{1,i,L}(c_{i,L} ,t)$ and $\gamma_{1,i,U}(c_{i,U} ,t)$, $i\in[1;n]$, define the normalized error $e_1(x_1(t), \gamma_1(t))$, the transformed error $\varepsilon_1(x_1(t), \gamma_1(t))$ and the diagonal matrix $\xi_1(x_1(t), \gamma_1(t))$ as
\begin{subequations} \label{eq:stage I}
   \begin{align}
    e_1(x_1(t), \gamma_1(t)) &= [e_{1,1}(x_{1,1},t), \ldots, e_{1,n}(x_{1,n},t)]^\top = (\gamma_{1,d} (t))^{-1} \left( 2x_1 - \gamma_{1,s} (t) \right),   \\
    \varepsilon_1(x_1(t), \gamma_1(t)) &= \big[\ln\left(\frac{1+e_{1,1}(x_{1,1},t)}{1-e_{1,1}(x_{1,1},t)}\right), \ldots, \ln\left(\frac{1+e_{1,n}(x_{1,n},t)}{1-e_{1,n}(x_{1,n},t)}\right) \big]^\top ,\\
    \xi_1(x_1(t), \gamma_1(t)) &= \frac{4 (\gamma_{1,d} (t))^{-1}}{1-e_1^\top(x_1(t), \gamma_1(t))e_1(x_1(t), \gamma_1(t))},
    \end{align} 
\end{subequations}
where, $\gamma_{1} := [\gamma_{1,1,U}, \gamma_{1,1,L}, \ldots, \gamma_{1,n,U}, \gamma_{1,n,L}]^\top$, $\gamma_{1,s} := [\gamma_{1,1,U} + \gamma_{1,1,L}, \ldots, \gamma_{1,n,U} + \gamma_{1,n,L}]^\top$ and $\gamma_{1,d} := \textsf{diag}(\gamma_{1,1,d},\ldots,\gamma_{1,n,d})$, with $\gamma_{1,i,d} = \gamma_{1,i,U} - \gamma_{1,i,L}$.

The intermediate control input $r_2(x_1(t), \gamma_1(t))$ is given by: 
\begin{equation*}
    r_2(x_1(t), \gamma_1(t)) = - \kappa_1\varepsilon_1(x_1(t), \gamma_1(t))\xi_1(x_1(t), \gamma_1(t)), \kappa_1 \in \R^+.
\end{equation*}

{\textbf{Stage $k$} ($k \in [2;N]$):} Given the reference vector $r_k(\overline{x}_{k-1}(t),\overline{\gamma}_{k-1}(t))$, we aim to design the subsequent intermediate control $r_{k+1}(\overline{x}_{k}(t), \overline{\gamma}_{k}(t))$ for the dynamics of $x_{k}$, ensuring that $x_{k}$ tracks the trajectory determined by $r_k(\overline{x}_{k-1}(t),\overline{\gamma}_{k-1}(t))$. Here, $\overline{\gamma}_{k}(t)$ is defined as $\overline{\gamma}_{k}(t) := [\gamma_{1}(t), \ldots, \gamma_{k}(t)]^\top$ with $\gamma_{k}(t):= [\gamma_{k,1,U}, \gamma_{k,1,L}, \ldots, \gamma_{k,n,U}, \gamma_{k,n,L}]^\top$.

This is done by enforcing exponentially narrowing constraint 
\begin{equation}\label{eqn:funnel}
\gamma_{k,i}(t) = (p_{k,i} - q_{k,i})e^{-\mu_{k,i}t} + q_{k,i} 
\end{equation}
such that 
\begin{align*}
    -\gamma_{k,i}(t) \leq (x_{k,i}-r_{k,i}) \leq \gamma_{k,i}(t) \ \ \forall i \in [1;n].
\end{align*}
Note that $|x_{k,i}(t=0) - r_{k,i}(t=0)| \leq p_{k,i}$.

Now define the normalized error $e_k(\overline{x}_{k}(t), \overline{\gamma}_{k}(t))$, the transformed error $\varepsilon_k(\overline{x}_{k}(t), \overline{\gamma}_{k}(t))$ and the diagonal matrix $\xi_k(\overline{x}_{k}(t), \overline{\gamma}_{k}(t))$ as
\begin{subequations} \label{eq:stage k}
    \begin{align}
    e_k(\overline{x}_{k}(t), \overline{\gamma}_{k}(t)) &= [e_{k,1}, \ldots, e_{k,n}]^\top = (\gamma_{k,d} (t))^{-1} \left(x_{k} - r_k(\overline{x}_{k-1}, t) \right), \\
    \varepsilon_k(\overline{x}_{k}(t), \overline{\gamma}_{k}(t)) &= \big[\ln\left(\frac{1+e_{k,1}}{1-e_{k,1}}\right), \ldots, \ln\left(\frac{1+e_{k,n}}{1-e_{k,n}}\right) \big]^\top, \\
    \xi_k(\overline{x}_{k}(t), \overline{\gamma}_{k}(t)) &= \frac{4 (\gamma_{k,d} (t))^{-1}}{1-e_k^\top(\overline{x}_{k}(t), \overline{\gamma}_{k}(t))e_k(\overline{x}_{k}(t), \overline{\gamma}_{k}(t))},
\end{align}
\end{subequations}
where $\gamma_{k,d} := \textsf{diag}(\gamma_{k,1,d},\ldots,\gamma_{k,n,d})$, with $\gamma_{k,i,d} = \frac{1}{2} (\gamma_{k,i,U} - \gamma_{k,i,L}) \nonumber, \forall i \in [1;n]$.

So, the intermediate control inputs $r_{k+1}(\overline{x}_{k}(t), \overline{\gamma}_{k}(t))$ to enforce the desired reach-avoid task is given by 
\begin{equation*}
    r_{k+1}(\overline{x}_{k}(t), \overline{\gamma}_{k}(t)) = - \kappa_k\varepsilon_k(\overline{x}_{k}(t), \overline{\gamma}_{k}(t))\xi_k(\overline{x}_{k}(t), \overline{\gamma}_{k}(t)), \kappa_k \in \R^+.
\end{equation*}
At the $N$-th stage, $r_{N+1}(\overline{x}_N(t), \overline{\gamma}_N(t))$ essentially serves as the actual control input $u(\overline{x}_N(t), \overline{\gamma}_N(t))$. It is important to note that $\overline{\gamma}_N(t)$ is a set of the spatiotemporal tubes in Equation \eqref{eqn:stt} for Stage 1 and all the subsequent exponentially narrowing functions in Equation \eqref{eqn:funnel} for Stages 2 through N. For ease of interpretation we use $\gamma(t) := \overline{\gamma}_N(t)$. Thus, the control input $u(\overline{x}_N(t), \gamma(t))$ is given by,
\begin{equation*}
    u(\overline{x}_N(t), \gamma(t)) = - \kappa_N\varepsilon_N(\overline{x}_N(t), \gamma(t))\xi_N(\overline{x}_N(t), \gamma(t)), \kappa_N \in \R^+.
\end{equation*}

Thus, we can design the control input to perform the reach-avoid specification for the system described in \eqref{eqn:sysdyn}.

\begin{theorem} \label{theorem_ras}
    Given an unknown nonlinear MIMO system in \eqref{eqn:sysdyn} satisfying assumptions \ref{assum:lip} and \ref{assum:pd}, a reach-avoid task, and corresponding STTs, if $y(0)$ is within the spatiotemporal tubes at time $t=0$, i.e., $\gamma_{1,i,L}(0) < y_i(0) < \gamma_{1,i,U}(0), \forall i \in [1;n]$, then the closed-form control strategies,
    \begin{subequations}\label{eqn:Control_ras}
     \begin{align}
        r_{k+1}(\overline{x}_{k}(t), \overline{\gamma}_{k}(t)) &= - \kappa_k\varepsilon_k(\overline{x}_{k}(t), \overline{\gamma}_{k}(t))\xi_k(\overline{x}_{k}(t), \overline{\gamma}_{k}(t)), k \in [1;N-1], \\
        u(\overline{x}_N(t), \gamma(t)) &= - \kappa_N\varepsilon_N(\overline{x}_N(t), \gamma(t))\xi_N(\overline{x}_N(t), \gamma(t)), 
    \end{align}    
    \end{subequations}
    will satisfy \eqref{eqn:ppc}, guiding the output $y(t)$ to the target $T$, while avoiding the unsafe set $U$. Here, $\varepsilon_k$, and $\xi_k$ are given in Equations \eqref{eq:stage I} and \eqref{eq:stage k}, respectively.
\end{theorem}

\begin{proof}
    The proof follows similar to that of Theorem 4.1 of \cite{das2024spatiotemporal}.
\end{proof}

{\begin{remark}    
It is important to note that the closed-form control law in Equation \eqref{eqn:Control_ras} is both approximation-free, being independent of the system dynamics, and of low complexity, eliminating the need for iterative numerical solvers or computationally expensive online optimization.
\end{remark}}

\subsection{Hybrid Control Policy}
\label{subsec4_2}
After decomposing the specification, we leverage the STT-based control approach introduced in the previous section to address these reach-avoid tasks. By combining the results, we design a hybrid control policy ensuring that the system satisfies the full specification defined by the NBA.

We provide a lemma that correlates a specific segment in the specification NBA $\mathcal{A}$, referred to as a triplet, with the result of Theorem \ref{theorem_ras}.

\begin{lemma}
    {For a triplet $\nu = (q, q', q'')$, decomposed into a reach-avoid task with the initial set $S = L^{-1}(\sigma(q,q'))$, target set $T = L^{-1}(\sigma(q',q'')),$ and unsafe set $U = Y \setminus \left( L^{-1}(\sigma(q',q')) \cup S \cup T \right)$, let the set of spatiotemporal tubes from Equation \eqref{eqn:stt} and all the exponentially narrowing functions in Equation \eqref{eqn:funnel} be $\gamma_{\nu}(t)$ and the control policy from Equation \eqref{eqn:Control_ras}(b) in Theorem \ref{theorem_ras} be defined as $u_\nu(\overline{x}_{N}(t), \gamma_{\nu}(t))$. Then there exists $t \in \mathbb{R}_0^+$ such that the trajectory $y_{y_0 u_v}$ of $\mathcal{S}$ starting from any initial state $y_0 \in S$ under the policy $u_\nu$ reaches the target region $T$, that is, $y_{y_0 u_v} \cap T \neq \emptyset$.}
\end{lemma}
\begin{proof}
    The proof follows similar to that of Theorem 4.1 of \cite{das2024spatiotemporal}.
\end{proof}

Given an NBA $\mathcal{A} = (\mathcal{Q}, \mathcal{Q}_0, \Pi, \delta, F)$ representing the desired properties, along with an accepting run $q$, and its associated finite state-run fragment $\overline{\text{q}}$, we define a finite state transition system to provide a switching mechanism for a hybrid control policy. 

The switching mechanism is described by a finite state transition system $ \mathfrak{S} = (\mathcal{Q}_s, \mathcal{Q}_{0s}, \Pi, \delta_s),$ where $\mathcal{Q}_{0s} = \mathcal{Q}_0$, $\mathcal{Q}_s = \mathcal{Q}_{0s} \cup \mathcal{P}^p (\overline{\text{q}})$, and transition relation $(q_s,\sigma,{q'}_s) \in \delta_s$ ($q_s \xrightarrow{\sigma}_{\mathfrak{S}} {q'}_s)$ is defined as: 

\begin{itemize}
    \item for all $q_s = q_0 \in \mathcal{Q}_{0s},$
    \begin{align}
        q_0 \xrightarrow{\sigma (q_0, q')}_{\mathfrak{S}} (q_0, q', q''), \text{where} \; q_0 \xrightarrow{\sigma (q_0, q')}_\mathcal{A} q'; \nonumber
    \end{align}
    \item for all $q_s = (q, q', q'') \in \mathcal{Q}_s \setminus \mathcal{Q}_{0s}$,
    \begin{align}
        (q,q',q'') &\xrightarrow{\sigma (q', q'') }_{\mathfrak{S}} (q',q'',q'''), \text{such that} \; q, q', q'', q''' \in \mathcal{Q}, \nonumber \\ 
        &\quad q' \xrightarrow{\sigma (q', q'')}_\mathcal{A} q'', \quad q'' \xrightarrow{\sigma (q'', q''')}_\mathcal{A} q'''. \nonumber
    \end{align}
\end{itemize}

{The hybrid controller defined over $Y \times \mathcal{Q}_s$ is given by:
\begin{align}\label{eqn:hybridcontroller}
    \tilde{\mathbf u}(\overline{x}_N(t), q_s) &:= u_{q_s'}(\overline{x}_N(t), \gamma_{q_s'}(t-t_{q_s})), \forall(q_s, L(x_1), q'_s) \in \delta_s, \ \forall t \in \R_0^+,
\end{align}
where for $q_s = (q,q',q'')$, $t_{q_s} \in \mathbb{R}^+_0$ is defined as $t_{q_s} := \min \{ t \in \mathbb{R}_0^+ | y_{y_0 \mathbf{\Tilde{u}}}(t) \in \sigma(q',q'')$.}

The following theorem states that the proposed hybrid controller in (\ref{eqn:hybridcontroller}), guarantees the satisfaction of the specification given by the language of $\mathcal{A}$.

\begin{theorem}
Consider an unknown system $\mathcal{S}$ in \eqref{eqn:sysdyn} satisfying assumptions \ref{assum:lip} and \ref{assum:pd} and an NBA $\mathcal{A}$ with a finite state-run fragment $\overline{\mathsf{q}}\in\mathcal{R}^p$ for some $p\in\Pi$, corresponding to an accepting run $\mathsf{q}$. The trajectory $y_{y_0\tilde{\mathbf u}}$ of $\mathcal S$ starting from $y_0\in L^{-1}(p)$ under the hybrid controller $\tilde{\mathbf u}$ in \eqref{eqn:hybridcontroller} satisfies the language of NBA $\mathcal{A}$, that is, $\sigma(y_{y_o\tilde{\mathbf u}})\models\mathcal{A}$. 
\end{theorem}

\begin{proof}
    Consider an accepting state run $\text{q} = (q_0^r, q_1^r, ... , q_{m_r}^r ,$ $ (q_0^s, q_1^s, ... , q_{m_s}^s)^\omega) \in \mathcal{Q}^\omega$ in $\mathcal{A}$ with $\sigma(q_0^r, q_1^r) = p$, for some $p \in \Pi$. Let the corresponding finite state-run be $\overline{\mathbf{q}} \in \mathcal{R}^p$, as defined in Subsection \ref{subsec_mapping}. If we employ a controller that switches between the reach-avoid control laws $u_\nu(x, t)$ for each $\nu = (q,q',q'') \in \mathcal{P}^p(\overline{\mathbf{q}})$ as established in Lemma 1, we can conclude that $\sigma(\text{q}) \in \mathcal{L}(\mathcal{A}).$ By applying the definition of the labeling function $L$, this implies that the trajectory $y_{y_0 \mathbf{\Tilde{u}}}$ of $\mathcal{S}$ starting from any initial state $y_0 \in L^{-1}(p)$ under the control policy $\mathbf{\Tilde{u}}$ in \eqref{eqn:hybridcontroller} satisfies $\sigma(y_{y_0 \mathbf{\Tilde{u}}}) \in \mathcal{L}(\mathcal{A})$. 
\end{proof}

\section{Case Studies}
\label{sec5}
In this section, we present two case studies to illustrate the efficacy of our results and their practical applicability. We applied our approach to two systems: $(i)$ a manipulator executing a pick-and-place operation and $(ii)$ an omnidirectional mobile robot performing a delivery task. 

\subsection{2-R Manipulator}
\label{subsec5_1}
\begin{figure}
    \centering
    \begin{minipage}{0.45\textwidth}
        \centering
        \includegraphics[page = {2},width=0.7\textwidth]{o_reg_automata.pdf} 
        \caption{2R Manipulator NBA}
        \label{2r_nba}
    \end{minipage}\hfill
    \begin{minipage}{0.45\textwidth}
        \centering
        \includegraphics[width=0.7\textwidth]{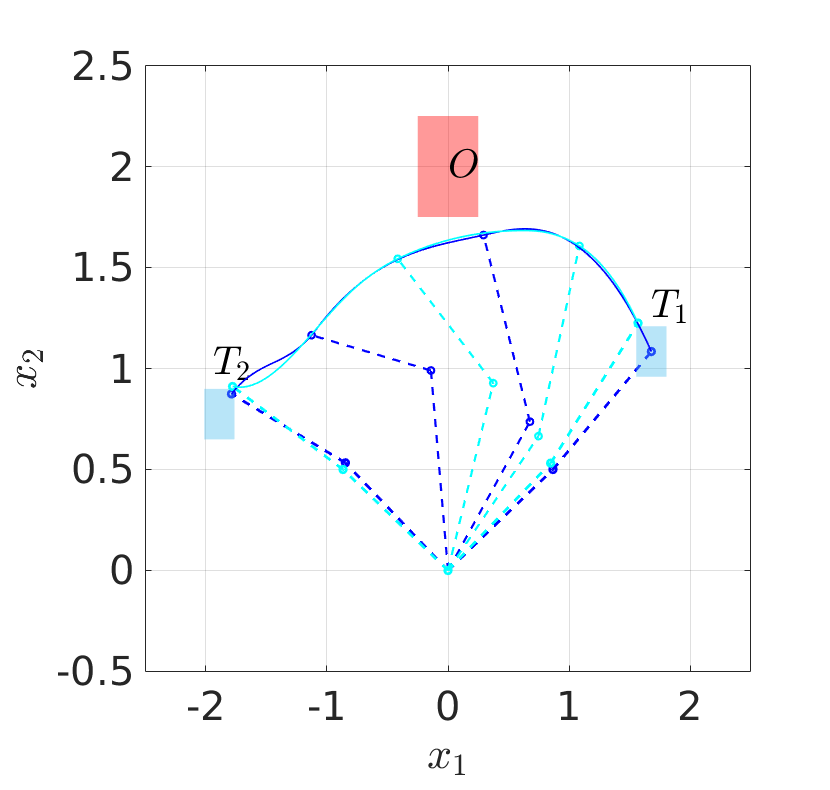} 
        \caption{Trajectory of the 2R manipulator}
        \label{2r_traj}
    \end{minipage}
\end{figure}

For the first case study, we consider a two-link 2R SCARA manipulator adapted from \cite{spong}. The following describes the model used:
\begin{gather}
    ml^2
    \begin{bmatrix}
        \frac{5}{3} + c_2 & \frac{1}{3} + \frac{1}{2}c_2 \\
        \frac{1}{2}c_2 & \frac{1}{3}
    \end{bmatrix}
    \begin{bmatrix}
        \Ddot{\theta}_1 \\
        \Ddot{\theta}_2
    \end{bmatrix} 
    +
    ml^2s_2
    \begin{bmatrix}
        -\frac{1}{2}\dot{\theta}_2^2 - \dot{\theta}_1\dot{\theta}_2\\
        \frac{1}{2}\dot{\theta}_2^2
    \end{bmatrix}
    \\
    + mgl
    \begin{bmatrix}
        \frac{3}{2}c_1 + \frac{1}{2}c_{12} \\
        \frac{1}{2}c_{12}
    \end{bmatrix}
    = 
    \begin{bmatrix}
        \tau_1(t) \\
        \tau_2(t)
    \end{bmatrix}
    + d(t), \nonumber
\end{gather}
where $m$ and $l$ are the mass and length of each link, $g$ is the acceleration due to gravity, $\tau_1(t), \tau_2(t)$ are the torque inputs at the joints, $d(t)$ is an unknown bounded disturbance, $c_1 = \cos{\theta_1}$, $c_2 = \cos{\theta_2}$, $s_2 = \sin{\theta_2}$, and $c_{12} = \cos{(\theta_1 + \theta_2)}$.

The manipulator was assigned the task of always eventually going to the targets $T_1 = [\pi /6 - \pi /18, \pi /6 + \pi / 18] \times [0.075 , 0.125]$ and $T_2 = [5\pi /6 - \pi /18, 5\pi /6 + \pi /18] \times$ $ [0.075 , 0.125]$, while always avoiding the obstacle $O = [0, \pi /18] \times [2 - \pi /10, 2 + \pi /10 ]$. 

{The set of atomic propositions is given by $\Pi = \{ p_0, p_1, p_2, p_3 \}$ with the labeling function $L:\R^2 \rightarrow \Pi$ defined as: $L(\theta_1, \theta_2) = p_0, \forall [\theta_1, \theta_2] \in O$, $L(\theta_1, \theta_2) = p_1, \forall [\theta_1, \theta_2] \in T_1$, $L(\theta_1, \theta_2) = p_2, \forall [\theta_1, \theta_2] \in T_2$, and $L(\theta_1, \theta_2) = p_3, \forall [\theta_1, \theta_2] \in \R^2 \setminus (O \cup T_1 \cup T_2)$.}
The objective is to compute a hybrid control policy ensuring the satisfaction of the specification given by the LTL formula 
$$\varphi = \Box \Diamond p_1 \wedge \Box \Diamond p_2 \wedge \Box \neg p_0,$$
or equivalently by the accepting language of the NBA $\mathcal{A}$ in Figure \ref{2r_nba}. From Figure \ref{2r_nba}, one can identify an accepting state-run $(q_0^*,q_1^*)^\omega.$
{We choose the accepting state-run $(q_0, q_1, q_0, q_1, q_0, q_1)$ which is decomposed as discussed in Section \ref{subsec_mapping}, to design a hybrid controller in Equation (\ref{eqn:hybridcontroller}).} Figure \ref{2r_traj} shows a trajectory of the system under the proposed hybrid control policy \eqref{eqn:hybridcontroller}. 

\begin{figure}
    \centering
    \begin{subfigure}[b]{0.3\textwidth}
        \centering
        \includegraphics[width= \textwidth]{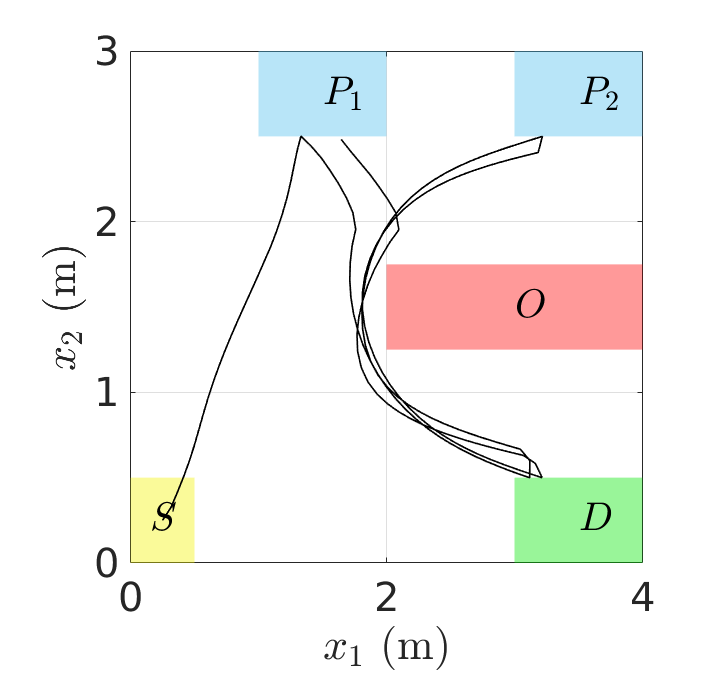}
        \label{fig:omni_traj}
    \end{subfigure}
    \begin{subfigure}[b]{0.5\textwidth}
        \centering
        \includegraphics[width= \textwidth]{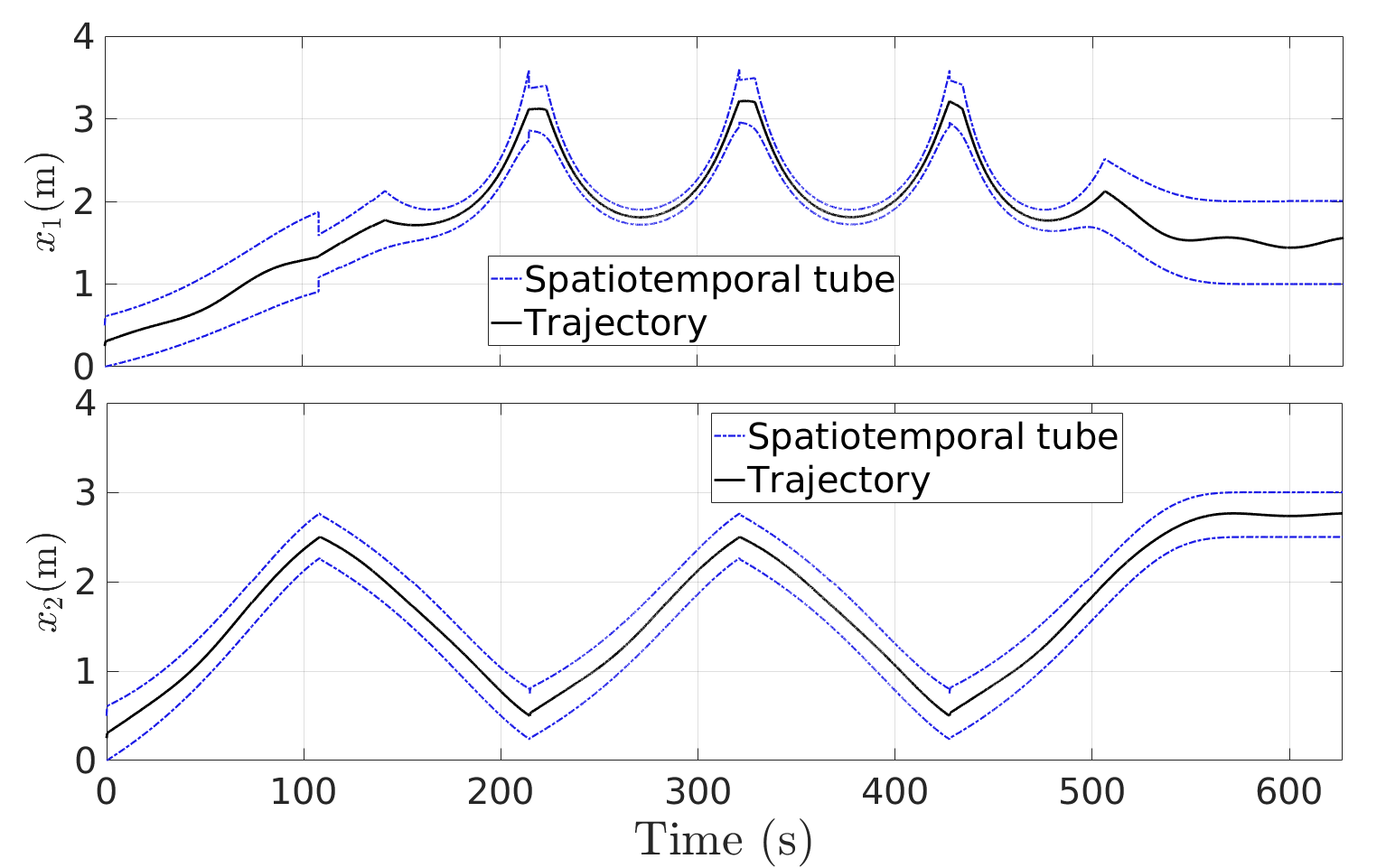}
        \label{fig:omni_tube}
    \end{subfigure}
    \caption{(a) Trajectory of the Mobile robot in 2D plane. (b) Tubes generated for each axis and the trajectory}
    \label{fig:omni}
\end{figure}

\subsection{Mobile robot}
\label{subsec_omni_ex}

For the second case study, we consider a three-wheeled omnidirectional mobile robot, adopted from \cite{liu2008omni} and defined as:
\begin{gather}
    \begin{bmatrix}
        \dot{x}_1 \\ \dot{x}_2 \\ \dot{x}_3
    \end{bmatrix}
    =
    \begin{bmatrix}
        cos(x_3) & -sin(x_3) & 0 \\
        sin(x_3) & cos(x_3) & 0 \\
        0 & 0 & 1
    \end{bmatrix}
    \begin{bmatrix}
        0 & -1 & L \\
        cos \frac{\pi}{6} & sin \frac{\pi}{6} & L \\
        -cos \frac{\pi}{6} & -sin \frac{\pi}{6} & L
    \end{bmatrix}^{-1}
    (B^T)^{-1}Ru,
\end{gather}
where the two states $x_1$ and $x_2$ indicate the robot position and $x_3$ indicates the orientation of the robot with respect to the $x_1$-axis. $R$ is the radius of the wheel, $B$ describes geometric constraints with $L$ is the radius of the robot body, and $u(t) \in \mathbb{R}^3$ are control inputs. 

The mobile robot was assigned the task of starting from location $S = [0, 0.5] \times [0, 0.5] \times [0, 2\pi]$, always eventually picking up objects from locations $P_1 = [1, 2] \times [2.5, 3] \times [0, 2\pi]$ and $P_2 = [3, 4] \times [2.5, 3] \times [0, 2\pi]$, and delivering them to location $D = [3, 4] \times [0, 0.5] \times [0, 2\pi]$, while always avoiding obstacle $O = [2, 4] \times [1.25, 1.75] \times [0, 2\pi]$, as shown in Figure \ref{fig:omni}.

{The set of atomic propositions is given by $\Pi = \{ p_0, p_1, p_2, p_3, p_4, p_5 \}$ with the labeling function $L:\R^3 \rightarrow \Pi$ defined as: 
$L(x_1, x_2, x_3) = p_0$, $\forall [x_1, x_2, x_3] \in S$, $L(x_1, x_2, x_3) = p_1, \forall [x_1, x_2, x_3] \in P_1$, $L(x_1, x_2, x_3) = p_2$, $\forall [x_1, x_2, x_3] \in P_2$, $L(x_1, x_2, x_3) = p_3, \forall [x_1, x_2, x_3] \in D$, $L(x_1, x_2, x_3) = p_4, \forall [x_1, x_2, x_3] \in O$, and $L(x_1, x_2, x_3) = p_5, \forall [x_1, x_2, x_3] \in \R^3 \setminus (S \cup P_1 \cup P_2 \cup D \cup O).$
}
The objective is to compute a hybrid control policy ensuring the satisfaction of the specification given by the LTL formula 
\begin{align*}
    \varphi &= p_0 \wedge \square(\lnot p_4) \wedge \square \lozenge p_1 \wedge \square \lozenge p_2 \wedge \square(p_1 \rightarrow ((\lnot p_2 \wedge \lnot p_4) \cup p_3)) \\
    & \quad \wedge \square (p_2 \rightarrow ((\lnot p_1 \wedge \lnot p_4) \cup p_3)).
\end{align*}
Figure \ref{fig:omni} shows the simulation results under the hybrid control policy \eqref{eqn:hybridcontroller}. 

\subsection{Discussion and Comparison}
To address $\omega$-regular objectives, various methods/tools have been proposed in the literature, including abstraction-based techniques \cite{tabuada2009verification,rungger2016scots,khaled2019pfaces,jagtap2017quest}, reachable set computation-based methods \cite{schurmann2020optimizing}, and control barrier function (CBF)-based discretization-free approaches \cite{jagtap2020formal}, \cite{anand2021compositional}. However, these techniques typically involve extensive numerical computations to design controllers that enforce complex specifications. Moreover, most of the existing approaches for LTL-controller synthesis, assume access to known system dynamics. Although \cite{NAHS} introduces a closed-form method to address $\omega$-regular objectives, it still relies on precise knowledge of system dynamics and is limited to a subset of $\omega$-regular objectives.
\\
In contrast, our proposed closed-form solution not only eliminates the need for known system dynamics but also extends to a general class of specifications expressed by $\omega$-regular languages. Furthermore, the proposed approach offers significant advantages in terms of scalability and robustness, thereby addressing critical limitations in prior synthesis techniques.

\section{Conclusion}
\label{conclusion}
In this paper, we introduced a discretization-free approach for synthesizing controllers for unknown control-affine MIMO systems to satisfy complex temporal properties expressed using $\omega$-regular languages or nondeterministic Büchi automata. The proposed framework decomposes the NBA into a sequence of reach-avoid tasks and leverages the spatiotemporal tubes (STT) framework to provide an approximation-free closed-form control policy that ensures the system's trajectory remains within the STTs, thereby guaranteeing the satisfaction of the reach-avoid task. Consequently, we propose an approximation-free closed-form hybrid control policy to satisfy the specifications expressed through the NBA. 
This makes our approach particularly suitable for real-time applications involving systems with unknown dynamics and complex specifications. Future work may explore the extension of this framework to handle input constraints.

\bibliographystyle{ieeetr}
\bibliography{sources}

\end{document}